\documentclass[letterpaper]{article} %
\usepackage{aaai24}  %
\usepackage{times}  %
\usepackage{helvet}  %
\usepackage{courier}  %
\usepackage[hyphens]{url}  %
\usepackage{graphicx} %
\usepackage{natbib}  %
\usepackage{caption} %
\usepackage{algorithm}
\usepackage{algorithmic}

\usepackage{newfloat}
\usepackage{listings}
\DeclareCaptionStyle{ruled}{labelfont=normalfont,labelsep=colon,strut=off} %
\floatstyle{ruled}
\newfloat{listing}{tb}{lst}{}
\floatname{listing}{Listing}
\nocopyright

\usepackage{amsmath,amssymb,amsthm}
\usepackage{makecell}
\usepackage{xcolor}  %
\usepackage{booktabs}
\usepackage{breqn}

\newtheorem{theorem}{Theorem}[section]

\newtheorem{definition}{Definition}[section]
\newtheorem{corollary}[theorem]{Corollary}
\newtheorem{lemma}[theorem]{Lemma}
\newcommand{\R}{\mathbb{R}}
\newcommand{\E}{\mathbb{E}}
\newcommand\norm[1]{\lVert#1\rVert}
\newcommand{\bs}[1]{\boldsymbol{#1}}

\definecolor{dccolor}{HTML}{bd7f0d}
\definecolor{mjccolor}{HTML}{0F97C7}

    \newcommand{\asw}{\mathsf{asw}}
    \newcommand{\sw}{\mathsf{sw}}
    \newcommand{\rev}{\mathsf{rev}}
    \newtheorem*{remark}{Remark}

\title{Automated Design of Affine Maximizer Mechanisms in Dynamic Settings}
\author {
    Michael~ Curry\equalcontrib\textsuperscript{\rm 1, \rm 2, \rm 4},
    Vinzenz~ Thoma\equalcontrib\textsuperscript{\rm 3, \rm 4},
    Darshan~ Chakrabarti\textsuperscript{\rm 5},
    Stephen McAleer\textsuperscript{\rm 6},
    Christian~ Kroer\textsuperscript{\rm 5},
    Tuomas~ Sandholm\textsuperscript{\rm 6, \rm 7},
    Niao~He\textsuperscript{\rm 3},
    Sven~Seuken\textsuperscript{\rm 2, \rm 4}
}
\affiliations {
    \textsuperscript{\rm 1}Harvard University\\
    \textsuperscript{\rm 2}University of Zurich\\
    \textsuperscript{\rm 3}ETH Zurich\\
    \textsuperscript{\rm 4}ETH AI Center\\
     \textsuperscript{\rm 5}Columbia University\\
      \textsuperscript{\rm 6}Carnegie Mellon University, Computer Science Department\\
      \textsuperscript{\rm 7}Optimized Markets, Strategy Robot, Strategic Machine\\
    curry@ifi.uzh.ch, vinzenz.thoma@ai.ethz.ch
}

\begin{document}

\maketitle

\begin{abstract}

Dynamic mechanism design is a challenging extension to ordinary mechanism design in which the mechanism designer must make a sequence of decisions over time in the face of possibly untruthful reports of participating agents.
Optimizing dynamic mechanisms for welfare is relatively well understood. However, there has been less work on optimizing for other goals (e.g. revenue), and without restrictive assumptions on valuations, it is remarkably challenging to characterize good mechanisms. Instead, we turn to \textit{automated mechanism design} to find mechanisms with good performance in specific problem instances.
We extend the class of affine maximizer mechanisms to MDPs where agents may untruthfully report their rewards. This extension results in a challenging bilevel optimization problem in which the upper problem involves choosing optimal mechanism parameters, and the lower problem involves solving the resulting MDP. 
Our approach can find truthful dynamic mechanisms that achieve strong performance on goals other than welfare, and can be applied to essentially any problem setting---without restrictions on valuations---for which RL can learn optimal policies.
\end{abstract}

\section{Introduction}

Dynamic mechanism design studies sequential decision-making problems, where decisions are based on the self-reported preferences of agents. A typical model is that the environment consists of a \textit{Markov decision process (MDP)}, and the mechanism controls the process given reported utilities by the agents. This has important applications, such as ad auctions or more generally online pricing (e.g. \citealt{Bergemann2019Dynamic}) but also problems of  decentralised decision making in RL (e.g. \citealt{Chang2020Decentralized}).

 Much work in dynamic mechanism design has focused on maximizing welfare~\cite{athey2013efficient,Nisan2007Algorithmic,lyu2022pessimism} subject to \textit{strategyproofness} (there should be no incentive for untruthful reports by agents).
 Some other work considers different goals, notably revenue~\cite{bergemann2010dynamic,kakade2013optimal,hajiaghayi2004adaptive,hajiaghayi2007automated} but needs to make restrictive assumptions about the space of agent types.
 Work on dynamic mechanism design, for general goals and for broad spaces of agent types, is much more limited. 

Dynamic mechanism design includes as a special case static mechanism design, and here the situation is similar.
To maximize welfare while ensuring strategyproofness, one can use the celebrated and well-understood Vickrey-Clarke-Groves (VCG) mechanism~\cite{Vickrey1961Counterspeculation,Clarke1971Multipart,Groves1973Incentives}.
For optimizing revenue, \citet{Myerson1981Optimal} completely settles the question under the restrictive assumption that agents' types are single-dimensional (essentially, that there is only one type of item up for sale);
 beyond this there has been little progress except for very specific problem instances~\cite{Yao2017Dominant}.
But for static settings,  \textit{automated mechanism design (AMD)}~\cite{Conitzer02Mechanism,Sandholm2003Automated,Curry2022Differentiable} has been used: this is a data-driven search through some class of mechanisms in order to find one that performs well while satisfying the constraints of strategyproofness and individual rationality. Automated mechanism design has in some cases found the highest-performing mechanisms known so far, and can recover optimal mechanisms in special cases where they are known~\cite{Duetting2019Optimal,ivanov2022optimal,Shen2019Automated}.

Given the successes of 
 automated mechanism design for static problems, it is surprising that its use for dynamic problems is relatively underexplored.

 \subsection{Our Contributions}
 Our present paper develops automated dynamic mechanism design techniques which can be applied to a very broad range of problems.
In particular, we consider mechanism design on general MDPs.
Our model works with many loss functions (including revenue but also other domain-specific loss functions), not just the easier goal of welfare, and we do not assume one-dimensional agent types.
Our assumptions are sufficiently general to capture essentially all of static multi-parameter mechanism design as a special case\footnote{If we restrict the MDP to have a single state, then we recover ordinary mechanism design, with each possible action corresponding to an outcome.}.

Optimal mechanism design over all possible mechanisms entails the very difficult problem of computing equilibria in imperfect information games, to understand whether or not any agent has any incentive to deviate from truthful reporting. Inspired by prior work for static mechanism design, we sidestep this issue by focusing on the class of affine maximizers (AMAs), which we define on MDPs.\footnote{The acronym AMA refers to ``affine maximizer auctions''. We consider affine maximizer \textit{mechanisms}, but we stick to the AMA acronym because it is widely used and to avoid confusion with AMM used to refer to ``automated market makers''.} These mechanisms are always strategyproof.

We justify this restriction in two ways.
First, as mentioned, our problem assumptions capture multi-parameter mechanism design---but finding optimal mechanisms in this setting has proven extremely difficult, so it makes sense to search within a more tractable class of mechanisms.
Second, our framing of the problem is broad enough that it captures problem instances where Roberts' theorem~\cite{Roberts1979Characterization} applies, which states that under general agent type spaces, \textit{only} affine maximizers can be strategyproof.

We frame the search for a high-performing \textbf{dynamic affine maximizer mechanism} as a \textbf{bilevel optimization problem}, where the outer problem consists of
choosing {weights} and (possibly state-dependent) {boosts}---the AMA parameters---to minimise a given loss function. 
The inner problem consists of learning to control an MDP to maximise affinely-transformed social welfare (the definition of an affine maximizer), given the weights, boosts, and agents' type reports (see equation \ref{eq:bilevel} below).
The derivatives of this inner problem do not exist at many points: however, we show that for the important case of revenue, the \textit{expected} loss over the distribution (with a continuous density) of agent valuations is differentiable.

We solve the inner problem via (possibly regularized) linear programming. We also propose a variety of ways to solve the outer problem: by grid search, by differentiating through the regularized LP, or by using zeroth-order methods to approximate the LP gradient. These latter two approaches explicitly or implicitly smooth the objective and avoid the problem of nonexistent derivatives. 
In experiments on several dynamic mechanism design settings, such as sequential auctions, task scheduling and navigating a gridworld, our approaches result in truthful mechanisms that outperform the VCG baseline.

\section{Related Work}
\subsection{Maximizing Welfare in Dynamic Mechanisms}

\citet{athey2013efficient} consider a dynamic mechanism design setting where agents update their beliefs over time, and where the goal is an efficient and budget-balanced outcome.
Parkes~\cite{ParkesChapterNisan2007algorithmic} describes a dynamic mechanism design setting where the focus is on agents who may arrive and depart at different periods. Both of these approaches simply assume an optimal policy is available.
More recently, \citet{lyu2022pessimism} presents a model for learning this policy in an offline RL setting; another work focuses on the online case~\cite{Lyu2022Learning}. \citet{Chang2020Decentralized} consider a dynamic VCG mechanism for decentralised RL where agents bid on the MDP transitions.  \citet{bergemann2010dynamic} presents a VCG-like ``dynamic pivot mechanism''.%
There are different modeling choices and goals in each of these approaches, but the common theme is that the allocation problem involves making decisions on some MDP after observing agent reports. All of these papers consider a dynamic analogue to the VCG mechanism, i.e. the mechanism designer acts according to the welfare-optimal policy and charges agents their externalities to ensure incentive compatibility.
This is in contrast to our work, where we are concerned with goals beyond welfare. 

\subsection{Dynamic Mechanism Design for Goals Other Than Welfare}

There is some existing work in this direction as well, with a particular focus on revenue. \citet{pavan2014dynamic} and \citet{kakade2013optimal} both consider cases where the private information about the value of the item is 1-dimensional. This allows for a Myerson-style analysis of the actual profit-maximizing mechanism. Other work considers optimal selling of items to agents who arrive and depart over time from the perspective of optimal stopping, but again considers single-parameter item valuations~\cite{hajiaghayi2004adaptive,hajiaghayi2007automated,Kleinberg2005MultipleChoice}. \citet{Bergemann2019Dynamic} surveys other results that make similar assumptions for tractability.
\citet{Pai2008Optimal} considers a similar setting and finds the optimal Bayes-Nash incentive compatible mechanism.
Still other work considers settings where the mechanism designer may update the mechanism online over time (e.g. by changing reserve prices)~\cite{Shen2017Reinforcement} and where bidders may even strategically attempt to manipulate the learning process~\cite{Amin2014Repeated}. 
Overall, these approaches are restrictive in terms of their assumed value model and frequently focus on analytic results, while our computational approach allows more general values and loss functions.

\subsection{Preference Elicitation from Multiple Agents and Multistage Mechanisms}

Another line of work considers iterative preference elicitation~\cite{Conen01:Preference,Sandholm06:Preference}: based on past agent reports, the mechanism can query what preference information it needs most. Existing approaches make use of an analogy to query learning~\cite{Zinkevich03:Polynomial,Blum2004Preference,Lahaie2004Applying}, or leverage machine learning~\cite{Soumalias2023Machine,Weissteiner2023Bayesian,Brero2019Machine}. 
In this context, \citet{Sandholm2007Automated} use automated mechanism design to first find good mechanisms (for general type spaces) and then convert them into multistage mechanisms.
While these approaches are typically focused on making a single final decision, but eliciting agent's preferences over multiple stages, our automated dynamic mechanism design approach is concerned with making a sequence of decisions over time, while eliciting preferences once.

\subsection{Static Automated Mechanism Design}

Due to the difficulty of analytically finding optimal mechanisms, a number of works have instead attempted to treat static mechanism design as a computational optimization problem, starting with~\citet{Conitzer02Mechanism} and~\citet{Sandholm2003Automated}, and learn good mechanisms from samples, starting with~\citet{Likhodedov2004Methods}. 
One line of work makes use of static affine maximizers and achieves good performance in multi-item multi-bidder auctions~\cite{Sandholm2015Automated, Curry2022Differentiable, Duan2023Scalable}.
There is also a line of learning theory research on choosing the AMA parameters given samples from the valuation distribution~\cite{Balcan2016Sample,Balcan2018General,Balcan2021How}.
Another direction is to start with a potentially non-strategyproof mechanism and iteratively modify it to improve strategyproofness. This is known as \textit{incremental mechanism design}~\cite{Conitzer07Incremental}. One line of work in this direction makes use of rich function approximators to learn mechanisms. \citet{Duetting2019Optimal} presents one influential direction, which uses neural networks to optimize revenue and a penalized loss to approximately enforce strategyproofness, with many followups~\cite{Curry2021Learning,Curry2020Certifying,ivanov2022optimal,Rahme2021Auction,Rahme2020Permutation}. \citet{Shen2019Automated} presents an alternative approach for single-bidder settings which can cope with broader classes of utility functions.
As mentioned above, the success of these techniques in static settings is our motivation to develop such approaches for the dynamic setting.

\section{Preliminaries}

Below, we describe our mechanism design problem. The nature of the problem and our mechanism design desiderata motivate our choice to restrict attention to affine maximizer mechanisms. We then describe in more detail how these mechanisms work in a dynamic setting.

\subsection{Formal Model of Problem Setting}

\paragraph{Environment and policy/allocation rule} Consider some MDP $\mathcal{M}=(S,A,P)$ (with reward not yet specified), where $S$ is a set of states, $A$ is a set of actions, and $P$ is a transition function.
There are $n$ agents, each with their own reward function $r_i: S\times A \rightarrow \R$ drawn from a distribution with density $f_i$. We emphasize that these agents are not themselves taking actions in the MDP---this is done by the mechanism. Their only choice will be which rewards to report to the mechanism.
We assume the mechanism designer wants to minimize some loss function (which will often be the negative of some objective to be maximized) $\mathcal{L}$ in expectation over the $f_i$, which in general depends on the agents' rewards and the chosen policy $\pi$ on the MDP. The latter corresponds to the allocation rule in traditional mechanism design.

\subsection{Mechanism Design Desiderata}

\paragraph{Incentive compatibility and payments} In order to achieve its goal, the mechanism will need access to the true $r_i(s,a)$.
However, in general, we should expect the agents to misreport their reward function if they think it will benefit them.
Thus, we will allow for some side payments to be made based on the MDP's solution and the agents' reports, in order to ensure that there is no incentive to misreport, that is, making the mechanism \textit{incentive compatible (IC)} or \textit{strategyproof}. (Such payments only exist for some choices of $\pi$.)
We assume agent utility is quasilinear, that is, positive payments just correspond to negative reward.

\paragraph{Individual rationality} We also want to guarantee individual rationality (IR), meaning that agents should not be charged so much that they receive negative utility and would be better off not participating in the mechanism. 

\begin{remark}
The mechanism designer's goal may also relate to the payments. For example, a canonical mechanism design goal is to maximize revenue---in our setting, choose a policy $\pi^*$ such that payments can be made as high as possible while still ensuring IC and IR.
\end{remark}

\subsection{Background on Affine Maximizers}

As motivated above our goal is to choose some $\pi$ such that IC/IR side payments can be constructed, while also performing as well as possible on the mechanism designer's higher-level objective.

Unlike prior work (e.g.,~\citealt{kakade2013optimal,bergemann2010dynamic}), we do not make any assumptions about the structure of the reward functions. Our setting is therefore general enough to incorporate many hard problems such as optimal multi-item mechanism design, and is thus at least as hard as those. Therefore hoping to get the truly best-performing $\pi$, even only in an infinite-sample/asymptotic sense, is too much. Thus it is appropriate to restrict attention to a more tractable class of mechanisms.

Also, our problem setting is general enough to include situations where a result known as Roberts' theorem applies~\cite{Roberts1979Characterization}. It states that under certain conditions (arbitrary rewards, at least 3 outcomes), the only allocation rules that can possibly have IC payments take the form of affine maximizers.

Below, we give the standard definition of affine maximizer mechanisms in terms of the allocation and payment rules, modified for our dynamic problem setting.

\begin{definition} [Affine maximizers]
    Given so-called \textit{weights} $w\in \R_{+}^n$ and boosts $b \in \R^{|S|\times |A|}$, a dynamic affine maximiser mechanism (AMA) takes reported reward functions $\bs{r}\in \R^{|S|\times |A|\times n}$ and returns a policy $\pi^*(w,b,\bs{r})$ on the MDP that maximizes the affine social welfare $\asw^{(\pi)}(w,b,\bs{r})$ where
  $$
    \asw^{(\pi)}(w,b,\bs{r}) = \E_{\pi}\left[\sum_{t=0}^T\left(\sum_{i=1}^n  w_i r_i(s_t,a_t)\right)+ b(s_t,a_t)\right]
    $$

    Defining $\asw(w,b,\bs{r}) =\asw^{(\pi^*(w,b,\bs{r}))}(w,b,\bs{r})$ as the maximum affine social welfare for reports $\bs{r}$, the resulting payment is then

\begin{dmath*}
\small
    \bs{p}_i(w,b,\bs{r})=  \frac{1}{w_i}\left(\asw^{(-i)}(w,b,\bs{r})-\left(\E_{\pi^*(w,b,\bs{r})}\left[\sum_{t=0}^T\left(\sum_{j\neq i}  w_j r_j(s_t,a_t)\right)+ b(s_t,a_t)\right] \right) \right) \label{eq:amapayment}
\end{dmath*}
   where 
  $\asw^{(-i)}(w,b,\bs{r})$ defined as $\max_{\pi} \E_{\pi}\left[\sum_{t=0}^T\left(\sum_{j\neq i}  w_j r_j(s_t,a_t)\right)+ b(s_t,a_t)\right]$
  is the maximum affine social welfare, when disregarding $i$.

\end{definition}
No matter the choice of weights and boosts, every resulting AMA is strategyproof ex-post---after learning the rewards of the other agents, reporting truthfully is a dominant strategy---and IR in expectation over the MDP trajectories. The proof can be found in Appendix \ref{app:amaisic}.

\section{Dynamic Mechanism Design as Bilevel Optimisation}
The problem of searching for a performant mechanism within the class of AMAs can be naturally formulated as stochastic bilevel optimization as follows\footnote{We constrain the policy to be in the \textit{set} of best responses, because there is a null set of possible $\bs{r}$ where $\pi^*$ is not unique. However, because it is a null set, the choice of $\pi^*$ does not influence the expected value in the outer problem so that it is well-defined.}:
\begin{dmath}
     \min_{w,b}  \mathbb{E}_{\bs{r}\sim f}[\mathcal{L}\left(\pi^*,w,b\right)] \text { s.t. } 
  \pi^* \in \arg \max _\pi \mathbb{E}_{s_t, a_t \sim \pi}\left[\sum_{t=0}^T  \left(\sum_{i=1}^n w_i r_i(s_t, a_t)\right)+b(s_t, a_t)\right] \forall \bs{r}
  \label{eq:bilevel}
\end{dmath}
Given the bilevel structure, we can think of the problem as a game between a leader and a follower.
\begin{itemize}
    \item The leader knows only the joint distribution $f=\prod_i f_i$ from which the set of reward functions are drawn, and chooses weights $w_i$ and boosts $b(s,a)$.
    \item For any draw of $r_i$ and the weights and boosts, the follower
    acts optimally (``best responds'') in the MDP according to the AMA objective.
\end{itemize}
To clarify (and to contrast with many models of mechanism design where the agents making reports are treated as followers): the leader and follower are only ``notional''.  In reality, there is only one mechanism designer. We nevertheless speak in terms of a ``leader'' and ``follower'' because in order for strategyproofness to be attained, some component of the system must successfully maximize affine social welfare, which is a goal distinct from the true goal of the mechanism designer.

In general, the problem above---bilevel optimization, with stochasticity in the leader's objective---is quite difficult. Indeed, the case of revenue-maximisation in one-round auctions with multiple goods, which is a very special case of our much more general problem, remains essentially unsolved beyond a few very simple special cases~\cite{Yao2017Dominant}.
Therefore finding a globally optimal solution to the above problem is too much to hope for, but we show that derivatives exist for important expected loss functions, which enables us to use gradient-based optimization techniques to find local optima.

We then consider three complementary methods for optimizing the AMA parameters: random grid search, zeroth-order methods to approximate the derivatives, and differentiation through a smoothed LP.

\subsection{Existence of Derivatives}
\label{sec:gradients}
So far, we have not made any assumptions about the loss function. A very natural desideratum would be that $\mathcal{L}(\pi^*(w,b,\bs{r}),w,b)$ is differentiable, so that we can perform stochastic gradient descent. However, in general this is not true. This is due to the fact that $\pi^*(w,b,\bs{r})$ is in general not even continuous in $w,b$---since the optimal policy on an MDP is always going to be deterministic---and therefore neither is $\mathcal{L}$. However, if $\mathcal{L}$ has a certain shape, which is the case for the loss functions we consider in this work, we can prove that the relaxed condition that $E_{\bs{r}}[\mathcal{L}(\pi^*(w,b,\bs{r}),w,b)]$ is differentiable.

\begin{theorem}
\label{thm:ldiff}
    Let $\mathcal{L}$ be a loss function for the problem in Equation \eqref{eq:bilevel} and assume it can be decomposed as follows \[\mathcal{L}(\pi^*,w,b)=\asw(w,b,\bs{r})+\sum_{k=1}^K a_k g_k(\pi^*(w,b,\bs{r}),\bs{r})\] where it holds for all $k$ that $g_k = \mathcal{O}(\norm{\bs{r}}_{\infty})$. Assume further that the support of $\bs{r}$ is compact or $f$ decays sufficicently quickly  such that $\E[\norm{\bs{r}}_{\infty}]$ exists. Then $E_{\bs{r}}[\mathcal{L}(\pi^*(w,b,\bs{r}),w,b)]$ is differentiable almost everywhere.
\end{theorem}
\begin{proof}[Proof Idea]
 The full proof is in Appendix \ref{app:proofdiff}. It consists of two parts. First we show that $\asw$ is Lipschitz continuous in $w,b$. This essentially follows because it is the maximum over a set of functions which are linear in $w,b$. For $g_k$ we argue that for any $w, b$, the set of rewards for which two alternative policies give the same affine social welfare, lie on a hyperplane. Changing the weights and boosts may change the optimal policy, and therefore may result in differing social welfare for some subset of the rewards. We would like to bound the probablility mass of the set of rewards where this can happen. Indeed, the rewards where the optimal policy may change, all lie in between the aforementioned hyperplanes. The distance between the hyperplanes depends on the change in $w$ and $b$, so the mass can be bounded---assuming the probability density decays sufficiently fast---which results in a local Lipschitz constant. Differentiability follows from Rademacher's theorem and holds equivalently for the sum of these terms because of the linearity of the expectation and derivative.
 \end{proof} %
 In Section \ref{sec:experiments} we will study two loss functions in particular: revenue and makespan. 
The revenue of a dynamic AMA mechanism is given by 
    \begin{dmath*}
        \mathsf{rev}(w,b,\bs{r})=\sum_{i=1}^n \bs{p}_i(w, b,\bs{r})
        =-\sum_{i=1}^n(\frac{1}{w_i}\asw(w,b,\bs{r}))+\sw^{(\pi^*(w,b,\bs{r}))}(\bs{r}) + \sum_{i=1}^n \frac{1}{w_i} \asw^{(-i)}(w,b,\bs{r})
    \end{dmath*}
Here $\sw$, the non-transformed social welfare, is just $\asw$ when weights are all 1 and boosts are all 0, i.e. $\sw^{(\pi)}=\E_{\pi}\left[\sum_{t=0}^T\sum_{i=1}^n r_i(s_t,a_t)\right]$.

It is easy to verify that revenue fulfills all assumptions of Theorem \ref{thm:ldiff}.

\begin{lemma}
    The expected revenue is differentiable almost everywhere.
\end{lemma}

\begin{proof}
    Revenue is a essentially a sum of the three terms $\asw,\sw$ and $\asw^{(-i)}$. The first $\asw$ is directly as stated in the assumptions of Theorem \ref{thm:ldiff}. For $\asw^{(-i)}$ the proof extends canonically as it equivalently is the maximal affine social welfare, albeit for a smaller set of agents. Last but not least $\sw^{(\pi^*(w,b,\bs{r}))}(\bs{r})$ is a function of only $(\pi^*(w,b,\bs{r}))$ and $\bs{r}$, for which it holds that
    \begin{align*}
    \sw^{(\pi^*(w,b,\bs{r}))}(\bs{r})&=\E_{\pi}\left[\sum_{t=0}^T\left(\sum_{i=1}^n   r_i(s_t,a_t)\right) \right]\\
    &\leq Tn \norm{\bs{r}}_{\infty}
    \end{align*}
    thereby concluding the proof.
\end{proof}

Once we have introduced the task scheduling problem we will analogously show that makespan also satisfies the assumptions of Theorem \ref{thm:ldiff}.

\subsection{Linear Programming Formulation}
\label{sec:lpformulation}
We now describe the linear-programming formulation for an MDP, which can be either infinite horizon or episodic with states partially ordered by time (that is, the time step is encoded in the state)~\cite{Altman1999Constrained}. In particular, we suppose the follower solves the following LP to find the optimal state-action occupancy measure $\nu_{\pi^*(w,b,\bs{r})}$,\footnote{This corresponds to the optimal policy by $\pi^*(a | s) = \frac{\nu_{\pi^*(w,b,\bs{r})}(s, a)}{\sum_a \nu_{\pi^*(w,b,\bs{r})}(s,a)}$.} which determines the revenue (or whichever loss function is chosen):
\begin{equation}
    \begin{split}
        	  &\max  \sum_{s \in S, a \in A} \left( \sum_i w_i r_i(s,a) 
	+ b(s,a) \right) \nu(s, a) \text{ s.t. }  \\
	 \ &\sum_{a \in A} \nu(s, a) = \sum_{s',a'} P(s|s', a') \nu(s', a')  + \mu_0(s) \quad \forall s \in S \\
	& \nu(s, a) \geq 0 \quad \forall s \in S, a \in A
	\label{eq:lp}
    \end{split}
\end{equation}
where $\mu_0$ denotes the initial state distribution. 

Note that if $\mathcal{L}$ were itself differentiable, we could take the gradient inside the expectation and apply the implicit function theorem to get $\nabla_{w, b} \mathbb{E}_{\boldsymbol{r}}\left[\mathcal{L}\left(\pi^*, w, b\right)\right]=\mathbb{E}_{\boldsymbol{r}}\left[\nabla_2 \mathcal{L}\left(\pi^*, w, b\right)+\nabla_{w, b} \pi^*(w, b, \boldsymbol{r}) \nabla_1 \mathcal{L}\left(\pi^*, w, b\right)\right]$. Then we could estimate the gradient of the expected value from a sum of gradients for different sampled $\bs{r}$. Indeed as we show in Appendix \ref{app:grad_asw}, for revenue we can even compute the partial derivatives of $\mathcal{L}$ with respect to $w,b$ analytically. However, because $\mathcal{L}$ is not differentiable (since $\pi^*$ is not), the Dominated Convergence Theorem does not hold and we cannot in fact exchange gradient and expected value. Instead we propose two alternatives to compute the gradient of the expected value. First using zeroth-order estimates and second introducing regularization, which makes the optimal policy in the inner problem unique $\mathcal{L}$ differentiable and thus actually allows us to estimate the gradient from samples. The pseudocode of our approach is given by Algorithm \ref{alg:algorithm}.

\subsubsection{Zeroth-order methods}
For bilevel problems, \citet{Sow2022Convergence} present a zeroth-order approach, which we adapt to our own setting.
The key observation of their approach is that the derivative of the leader's objective is the sum of two partial derivatives.
One of these, the derivative of the leader's objective with respect to their own solution, is usually easy to evaluate.
The other requires differentiating through the follower's best-response map and inverting a Jacobian, which is challenging---and this second portion can be separately estimated using zeroth-order perturbations.

Zeroth-order methods, due to the use of random perturbations, also implicitly smooth the function~\cite{duchi2012randomized}, so that when using them there is no further need for regularization to ensure that derivatives can be estimated. \footnote{They have been proposed in certain related (static) settings~\cite{Bichler2021Learning,Martin2022Finding} to cope with similar problems related to nonexistence of derivatives.} %

\subsubsection{Regularized linear program}
As an alternative method, we propose regularizing our problem to get a smooth surrogate of the derivative. 
We add a small entropy regularizer to the follower's objective. \footnote{A technique equivalent to such regularization has also been used to deal with a similar issue in first-order computation of equilibria of non-truthful static auctions~\cite{Kohring2023Enabling}.}
The result is now a convex program (see \ref{eq:smoothlp} in the appendix): 
When the objective is strongly convex, the solution map $\nu_r^*:(w,b)\mapsto \nu_{\pi^*(w,b,\bs{r})}$ is smooth.
Therefore, derivatives of revenue can now be estimated from the gradients at sampled type profiles, calculated using reverse-mode automatic differentiation~\cite{Agrawal2020Differentiating,Agrawal2019Differentiable}.

Adding a small amount of regularisation does not change the follower's problem significantly. As has been shown by \citet{Weed2018Explicit}, the distance between the solutions to the regularized and unregularized inner problem decays exponentially fast in the regularisation constant $\alpha$, for sufficiently small $\alpha$.

Under certain assumptions we claim this convergence translates to the outer problem, so that choosing a small $\alpha$ ensures the objective is not disturbed too much.

\begin{theorem} [Pointwise convergence of regularised loss]
\label{thm:regconv}
Assume that $\mathcal{L}(w,b,\bs{r})$ can be represented as an inner product between a vector $\bs{q}$ in $\R^{|S|\times|R|}$ and the state-action occupancy measure $\nu$, such that $\norm{q}_{\infty}\leq \mathcal{O}(\norm{r}_{\infty})$.
    Let $L_{\alpha}(w,b,\bs{r})$ denote the loss achieved with the optimal policy $\pi^{\alpha}(w,b,\bs{r})$ for the regularised problem (see equation \ref{eq:smoothlp} in Appendix \ref{app:reglp}). Assume further that $\mathbb{E}[\norm{\bs{r}}_\infty]$ exists, then 
    \[\lim_{\alpha\rightarrow 0}\E\left[\mathcal{L}_{\alpha}(w,b,\bs{r})\right]  = \E\left[\mathcal{L}(w,b,\bs{r})\right]\]
\end{theorem} \
The proof of the Theorem can be found in Appendix \ref{app:reg_rev}. Given that both affine social welfare and social welfare can be represented as an inner product of $\nu$, Theorem \ref{thm:regconv} implies the convergence of regularised revenue. Similarly, we will later show the same holds for makespan. 
\begin{corollary}
 Let $\rev_{\alpha}(w,b)$ denote the revenue achieved with the optimal policy $\pi^{\alpha}(w,b,\bs{r})$. Assume that $\mathbb{E}[\norm{\bs{r}}_\infty]$ exists, then 
    \[\lim_{\alpha\rightarrow 0} \rev_{\alpha}(w,b) = \rev(w,b)\]
\end{corollary}

\begin{algorithm}[tb]
\caption{Gradient-based Dynamic Mechanism Design}
\label{alg:algorithm}
\textbf{Input}: MDP $\mathcal M$, number of agents $n$, loss $\mathcal{L}$, number of initialisations $m$
\begin{algorithmic}[1] %
\STATE \textbf{Initialize:} $w\in \mathbb{R}^n$, $b \in \mathbb{R}^{|S|\times|A|}$
\STATE $(w_i,b_i)_{1\leq i\leq m}$ $\gets$ \textbf{grid search}$(\mathcal{M},\mathcal{L})$ \\
\COMMENT{Multiple starting points to avoid local minima}
\FOR{\texttt{num_iterations}}
    \FOR{$i = 1$ \TO $m$}
        \STATE $(w_i,b_i)$ $\gets$ $(w_i,b_i) + \gamma$ \textbf{gradientstep}$(\mathcal{M},w_i,b_i)$ \COMMENT{Zeroth or first order gradient estimate}
    \ENDFOR
\ENDFOR
\STATE \textbf{return} $\arg\min_{i\in \{1,\dots,m\}} \mathbb{E}[\mathcal{L}(w_i,b_i)]$
\end{algorithmic}
\end{algorithm}

\section{Experiments}
\label{sec:experiments}
\subsection{Methods}
We optimize AMAs in dynamic mechanism design settings on tabular MDPs. We compare three different mechanism design settings with different reward distributions. Our optimization methods consist of a na{\"i}ve grid search baseline, and two gradient-based methods (using either zero-order or regularized gradient estimates). The code we used to run these experiments is available at https://github.com/VnznzT/dynamicAMA.
\subsection{Implementation Details}

\paragraph{Grid search} As a methodological benchmark, we implement a na{\"i}ve grid search algorithm. We sample 10000 different weights and boosts from a Sobol sequence~\cite{Sobol1967distribution}. To assess the expected performance of each draw, we solve the associated dual linear program (LP), as detailed in Equation \eqref{eq:lp}, for 2000 randomly sampled reward profiles.

\paragraph{Zeroth-order methods} We estimate derivatives using 20 perturbations, sampled from a Gaussian distribution with standard deviation $0.05$ per estimate on 20 sampled type profiles. We use a learning rate of $0.1$.

\paragraph{Regularized LP.} We compute derivatives with respect to social welfare using the regularized LP, with smoothing parameter $10^{-2}$ except where mentioned. We solve the regularized program using MOSEK~\cite{mosek} and use the DiffOpt package within JuMP~\cite{Lubin2023} to differentiate. For the affine social welfare terms in the expected revenue we use the partial derivatives given in Appendix \ref{app:grad_asw}. For each stochastic gradient step, we sample 20 type profiles and optimize with learning rate $10^{-2}$.

In all cases, when evaluating the objective, we sample 10000 type profiles and do not use regularization. Thus at evaluation time, the LP solution is exactly correct, ensuring strategyproofness.
Computational details and hyperparameters are described in appendix \ref{app:hparam}. 

\subsection{Environments and Results}
\paragraph{Sequential Sales}

\begin{table}
\small
\centering
\begin{tabular}{@{}rrrrrrr@{}}
\toprule
\multicolumn{1}{l}{n} &
  \multicolumn{1}{l}{m} &
  \multicolumn{1}{l}{VCG} &
  \multicolumn{1}{l}{Grid} &
  \multicolumn{1}{l}{0-order} &
  \multicolumn{1}{l}{reg. LP} &
  \multicolumn{1}{l}{Best imp.}\\ \midrule
2 & 2 & 0.0000 & 0.4902 & \textbf{0.4939} & 0.4893          & \multicolumn{1}{r}{N/A}\\
3 & 2 & 0.4999 & 0.6079          & \textbf{0.6777} & 0.6755          & \multicolumn{1}{r}{35.56\%}\\
4 & 2 & 0.8000 & 0.7977          & \textbf{0.8783} & 0.8328          & \multicolumn{1}{r}{9.79\%}\\
5 & 2 & 0.9996 & 1.0050          & \textbf{1.0078} & 0.9967          & \multicolumn{1}{r}{0.83\%}\\
2 & 3 & 0.0000 & \textbf{0.4715} & 0.3825          & 0.4168          & \multicolumn{1}{r}{N/A}\\
3 & 3 & 0.0000 & 0.6107          & 0.6446          & \textbf{0.7240} & \multicolumn{1}{r}{N/A}\\
4 & 3 & 0.6007 & 0.7323          & 0.8875          & \textbf{0.9104} & \multicolumn{1}{r}{51.54\%}\\ 
5                    & 3                    & 0.9988               & 0.7142          & 1.0631          & \textbf{1.0743} & \multicolumn{1}{r}{7.56\%}\\ \bottomrule
\end{tabular}
\caption{Results for optimizing auction revenue in a sequential sales setting ($n$ agents, $m$ sales) with symmetric uniformly-distributed types. Standard errors were $<0.007$. Runtime was $<31$ hours for grid search, $<0.2$ hours for zeroth-order and $<1.1$ hours for first order.} 
\label{tab:revenue}
\end{table}

We begin with a simple setting in which identical items are sold sequentially to unit-demand bidders. The states consist of a record of who has received the item; the allowed actions are to sell the item to some bidder, or to no one. The welfare-maximizing mechanism thus involves giving the items to the highest bidder, but by altering the boosts, revenue can be increased by sometimes withholding the item. We consider a distribution of type profiles drawn uniformly from $[0,1]$, with results in Table~\ref{tab:revenue}. 

We observe that optimizing the boosts can consistently improve performance compared to VCG, especially when there are no tight supply constraints. Intuitively, if there is a large supply of goods, VCG revenue should be low, as agents do not cause much externality on other agents. By setting boosts to effectively withhold goods (like a reserve price), revenue can be increased.
\begin{table}[t]
\small
\centering
\begin{tabular}{@{}rrrrrrr@{}}
\toprule
\multicolumn{1}{l}{n} &
  \multicolumn{1}{l}{m} &
  \multicolumn{1}{l}{VCG} &
  \multicolumn{1}{l}{Grid} &
  \multicolumn{1}{l}{0-order} &
  \multicolumn{1}{l}{reg. LP} &
  \makecell {Best imp.} \\ \midrule
2 & 2 & 0.0000 & 0.3327          & 0.3302          & \textbf{0.3665} & \multicolumn{1}{r}{N/A}                     \\
3 & 2 & 0.2348 & 0.3217          & \textbf{0.4123} & 0.4116          & \multicolumn{1}{r}{75.55\%}                 \\
4 & 2 & 0.3089 & 0.3328          & 0.4021          & \textbf{0.4508} & \multicolumn{1}{r}{45.95\%}                 \\
5 & 2 & 0.3350 & 0.3355          & 0.4348          & \textbf{0.4645} & \multicolumn{1}{r}{38.65\%}                 \\
2 & 3 & 0.0000 & \textbf{0.3208} & 0.2544          & 0.3202          & \multicolumn{1}{r}{N/A} \\
3 & 3 & 0.0000 & 0.3304          & 0.3556          & \textbf{0.4354} & \multicolumn{1}{r}{N/A} \\
4 & 3 & 0.2276 & 0.3431          & 0.4263          & \textbf{0.4751} & \multicolumn{1}{r}{108.77\%}                \\ 
5 & 3 & 0.3193 & 0.2713          & 0.3770          & \textbf{0.4976} & \multicolumn{1}{r}{55.85\%}  \\ \bottomrule
\end{tabular}
\caption{Results for optimizing auction revenue in a sequential sales setting ($n$ agents, $m$ sales) with asymmetric uniformly-distributed types. Standard errors were $<0.004$. Runtime was $<32$ hours for grid search, $<0.2$ hours for zeroth-order and $<1.1$ hours for first order.}
\label{tab:revenueasym}
\end{table}

We also consider a distribution where agent $i$'s type is uniformly distributed on $[0,i]$. In this setting, we also allow the bidder weights to vary, with results in Table \ref{tab:revenueasym}. Again, we observe improved performance by optimizing the AMA parameters.

In both settings, the gradient-based approaches generally outperform the random grid search both in terms of results and runtime, in particular for the larger settings, which makes sense considering the high number of dimensions.

\paragraph{Dynamic truthful task scheduling}
\begin{table}[]
\small
\centering
\begin{tabular}{@{}rrrrrrr@{}}
\toprule
\multicolumn{1}{l}{n} &
  \multicolumn{1}{l}{m} &
  \multicolumn{1}{l}{VCG} &
  \multicolumn{1}{l}{Grid} &
  \multicolumn{1}{l}{0-order} &
  \multicolumn{1}{l}{reg. LP} &
  \multicolumn{1}{l}{Best imp.} \\ \midrule
2 & 4 & 1.0336 & 1.0326    & \textbf{0.9132} & 0.9288    & -11.65\%             \\
2 & 5 & 1.0116 & 1.0142    & \textbf{0.8820} & 0.9286    & -12.81\%             \\
3 & 4 & 0.6111 & 0.6196    & \textbf{0.5967} & 0.6184    & -2.36\%              \\
3 & 5 & 0.5662 & 0.5632    & \textbf{0.5429} & 0.5538    & -4.12\%              \\ \bottomrule
\end{tabular}
\caption{Results for minimizing makespan in the dynamic truthful task scheduling setting ($n$ agents, $m$ sales) with symmetric uniformly-distributed types. Standard errors were $<0.02$. Runtime was $<16$ hours for grid search, $<0.1$ hours for zeroth-order and $<4.1$ hours for first order. AMA outperforms VCG because the makespan is smaller.}
\label{tab:makespan}
\end{table}

\begin{table}[t]
\small
\centering
\begin{tabular}{@{}rrrrrrr@{}}
\toprule
\multicolumn{1}{l}{n} &
  \multicolumn{1}{l}{m} &
  \multicolumn{1}{l}{VCG} &
  \multicolumn{1}{l}{Grid} &
  \multicolumn{1}{l}{0-order} &
  \multicolumn{1}{l}{reg. LP} &
  \multicolumn{1}{l}{Best imp.} \\ \midrule
2 & 4 & 1.8312 & 1.8260    & \textbf{1.5978} & 1.6121    & -12.75\%             \\
2 & 5 & 1.9651 & 1.9837    & \textbf{1.6347} & 1.6886    & -16.81\%             \\
3 & 4 & 1.4299 & 1.4426    & \textbf{1.3242} & 1.3243    & -7.39\%              \\
3 & 5 & 1.4644 & 1.4729    & \textbf{1.3256} & 1.3629    & -9.48\%              \\ \bottomrule
\end{tabular}
\caption{Results for minimizing makespan in the dynamic truthful task scheduling setting ($n$ agents, $m$ sales) with asymmetric uniformly-distributed types. Standard errors were $<0.03$. Runtime was $<17$ hours for grid search, $<0.1$ hours for zeroth-order and $<4$ hours for first order. AMA outperforms VCG because the makespan is smaller.}
\label{tab:makespanasym}
\end{table}

The next setting we consider is a dynamic version of the classic truthful task scheduling problem \cite{Nisan2001Algorithmic}. In the static problem, workers report the time it takes them to complete certain tasks. A mechanism then has to assign the tasks and give payments\footnote{This in contrast to the auction environment where the mechanism could charge payments, because here agents suffer costs for which they need to be compensated.} to incentivize truthful reports with the goal of minimizing the makespan of all jobs.
We formulate a dynamic version of the truthful task scheduling problem. There are $n$ workers and $T$ tasks. Each round, one of the tasks arrives and has to be assigned. Each worker has a cost vector $\theta_i = (t_{i,1},\dots,t_{i,T})$ distributed according to some prior $f_i$, which consists of the times they take to finish each task. Each round $\tau$ the mechanism designer takes an allocative action $\bs{x}_\tau\in \{0,1\}^{n}$, s.t. $\sum_{i=1}^n \boldsymbol{x}_{\tau,i} =1$. At time $\tau$, this causes agent $i$ to receive reward $r_{i,\tau}=- \boldsymbol{x}_{\tau,i} t_{i,\tau}$.

The objective of the leader is to minimise total makespan. For this define $\tilde{t}_{i}$ as the time $i$ has already worked on its jobs, when the last task has been assigned in round $T$. The leader's loss is then given by 
$\max_{i\in \{1,\dots,n\}} \left(\left(\sum_{\tau=1}^T \boldsymbol{x}_{\tau,i} t_{i,\tau}\right)-\tilde{t}_i\right)$.

In order to justify using our approach for minimizing makespan, we need to show that Theorems \ref{thm:ldiff} and \ref{thm:regconv} apply. To see that makespan fulfills the assumption of Theorem \ref{thm:ldiff}, notice that it is a function of only $\bs{t}$ (which is $-\bs{r}$) and $\pi^*(w,b,\bs{r})$ and that the total makespan is always bound by $T\norm{\bs{r}}_{\infty}$. Therefore we can conclude that it is differentiable almost everywhere. Further notice that the makespan of a single agent (i.e. without the $\max$ operator in front) fulfills the assumptions of Theorem \ref{thm:regconv}. Let the maximum difference across all agent-specific makespans between the unregularised and regularised optimal policy be $\epsilon$, then the difference in the total makespan can be at most $2\epsilon$. But $\epsilon \rightarrow 0$ according to Theorem \ref{thm:regconv} which shows the makespan under the regularized optimal policies will converge to the correct makespan as $\alpha \rightarrow 0$.

As a benchmark, we consider a dynamic VCG mechanism that chooses the solution, which minimizes the total work done by the agents---an objective which is not the same as minimizing  makespan. (It can been shown, however, that VCG is an $n$-approximation of the optimal static mechanism~\cite{Nisan2001Algorithmic}.)

For a participant-symmetric valuation distribution (uniform on $[0,3]$), results are in Table \ref{tab:makespan}.  Across all environments, we see an improvement in makespan for the optimized AMAs over the VCG mechanism. 
We also consider an asymmetric distribution with disutilities distributed on $[0, 3i]$ for bidder $i$, with results in Table \ref{tab:makespanasym}.  Across all environments, we see a significant improvement in makespan for the best AMAs. In both settings, gradient-based approaches outperform the na{\"i}ve grid search both in terms of results and runtime, as was observed in the sequential auction environments.

\begin{table}[]
\small
\centering
\begin{tabular}{@{}rrrrrrr@{}}
\toprule
\multicolumn{1}{l}{n} &
  \multicolumn{1}{l}{m} &
  \multicolumn{1}{l}{VCG} &
  \multicolumn{1}{l}{Grid} &
  \multicolumn{1}{l}{0-order} &
  \multicolumn{1}{l}{reg. LP} &
  \multicolumn{1}{l}{Best imp.} \\ \midrule
2 & 3 & 0.7547 & 1.0607          & 1.3486 & \textbf{1.5464} & 104.90\% \\
3 & 3 & 1.2812 & 1.4348          & 1.5575 & \textbf{1.9251} & 50.25\%  \\
4 & 3 & 1.8683 & 1.8729          & 1.8853 & \textbf{2.3009} & 23.15\%  \\
5 & 3 & 2.3563 & 2.3689          & 1.9951 & \textbf{2.5989} & 10.30\%  \\
2 & 4 & 1.0402 & {1.0446} & 1.4935 & \textbf{1.6134} & 55.11\%  \\
3 & 4 & 1.4898 & 1.4904          & 1.6821 & \textbf{2.0171} & 35.40\%  \\
4 & 4 & 1.8610 & 1.8757          & 1.9427 & \textbf{2.3413} & 25.81\%  \\
5 & 4 & 2.2472 & 2.2133          & 2.0256 & \textbf{2.5911} & 15.30\%         \\ 
\bottomrule
\end{tabular}
\caption{Results for minimizing makespan in the gridworld environment. Standard errors were $<0.05$. Runtime was $<30$ hours for grid search, $<0.1$ hours for zeroth-order and $<0.1$ hours for first order.}
\label{tab:gridworldrev}
\end{table}

\paragraph{Navigating a grid with multiple tasks}
One of the most canonical environments in RL is the \textit{gridworld}, where an agent deterministically navigates a two-dimensional grid with rewards for reaching certain states~\cite{Sutton2018Reinforcement}. We consider the following variant: the mechanism moves ("up","down","left","right") in a grid with $n$ agents observing the trajectories. It starts in state $s_0$. Each agent $i$ draws a goal state $s_i \in S\setminus \{s_0\}$ and a reward $r_i \sim \mathcal{U}(0,1)$, which they receive when the mechanism reaches $s_i$. Given $(w,b)$, the mechanism finds a policy to maximize  affine social welfare $\pi^* \in \arg\max_{\pi} \E_{\pi}\left[\sum_{t=0}^\infty \gamma^t \left(\sum_{i=1}^n w_i r_i(s_t,a_t)+ b(s_t,a_t)\right)\right]$
 where $\gamma$ is the discount factor to account for the infinite time horizon. 
 One can interpret this environment as an auctioneer navigating an environment with different replenishing goods which agents wish to collect. The agents now bid to influence trajectories, and the auctioneer tries to maximize revenue. The results are in Table~\ref{tab:gridworldrev} and show that we can increase revenue by at least 10\% in every setting considered. Moreover, the gradient-based approaches are once more much faster and achieve better results than grid search. We further observed that the optimized boosts correspond to a preference of the mechanism for staying close to $s_0$. This can be interpreted roughly as a reserve price. 

Overall, we conclude that searching in the class of dynamic AMAs produces performant mechanisms in a variety of settings, which consistently outperform dynamic VCG. On a methodological level, a naïve grid search can perform well in settings with small dimensionality, but cannot scale up and takes a magnitude more runtime compared to our gradient-based approaches, which perform well across all settings. The best choice between zeroth-order and regularized LP depends on the mechanism design setting.

\section{Conclusion}
In this paper, we have proposed an approach for automated dynamic mechanism design. In contrast to earlier work, this formulation allows for a wide array of possible objectives (not just maximising social welfare) and works without strong restrictions on the type space. In principle, it captures essentially all problems of static mechanism design as a special case.
By focusing on the class of AMAs, we can frame the problem as stochastic bilevel optimization, where the mechanism designer acting in the outer problem chooses parameters to maximize their objective in expectation over possible rewards and the inner problem consists of optimally solving the MDP.

For the most prominent objective in mechanism design---expected revenue---we have further proven differentiability, which allows for gradient-based optimisation approaches to converge to locally optimal mechanisms. Because we restrict to the class of AMAs, all these mechanisms are guaranteed to be exactly IC and IR. 
To solve the bilevel problem, we have presented randomized grid search, as well as a zeroth and first order gradient-based algorithm to find well-performing mechanisms, which can beat the benchmark dynamic VCG mechanism across a broad range of environments we consider. The gradient-based methods we propose also consistently outperform na{\"i}ve grid search, which suffers from the curse of dimensionality.

The method we have presented is appropriate for any problem that can be formulated as controlling an MDP in the face of possibly untruthful preferences. This covers a wide range of interesting scenarios. In particular, the use of affine maximizers and the bilevel problem formulation are applicable to a broader range of settings, including those beyond the reach of tabular methods. Future work could apply deep RL for both the leader and follower, enabling scaling to significantly larger and more complicated problems, or apply our techniques to novel mechanism design settings.

\section*{Acknowledgments}

This paper is part of a project that has received funding from the European Research Council (ERC) under the European Union’s Horizon 2020 research and innovation program (Grant agreement No. 805542). N.H. is supported by ETH research grant and Swiss National Science Foundation (SNSF) Project Funding No. 200021-207343. V.T. is supported by an ETH AI Center Doctoral Fellowship. D.C. was supported by the National Science Foundation Graduate Research Fellowship Program under award number DGE-2036197. C.K. was supported by the Office of Naval Research awards N00014-22-1-2530 and N00014-23-1-2374, and the National Science Foundation awards IIS-2147361 and IIS-2238960. T.S. is supported by the Vannevar Bush Faculty Fellowship ONR N00014-23-1-2876, National Science Foundation grants RI-2312342 and
RI-1901403, ARO award W911NF2210266, and NIH award A240108S001.
 \bibliography{references}

\onecolumn
\appendix
\section{Affine maximizers are incentive compatible}
\label{app:amaisic}

\begin{theorem}
For any choice of fixed weights and boosts, affine maximizers are incentive compatible.
\end{theorem}

\begin{proof}
(Following the standard structure for truthfulness of VCG in the static case.) Consider player $i$ with true reward function $r_i$ and reported rewards $\tilde{r}_i$, where other players have rewards $r_{-i}$.
Recall that the payments charged for each agent can be written as

\begin{align*}
    p_i(\bs{r})&=  \frac{1}{w_i}\left(\asw^{(-i)}(w,b,\bs{r})-\left(\E_{\pi^*(w,b,\bs{r})}\left[\sum_{t=0}^T\left(\sum_{j\neq i}  w_j r_j(s_t,a_t)\right)+ b(s_t,a_t)\right] \right) \right) \\
 &=\frac{1}{w_i}\left(\asw^{(-i)}(w,b,\bs{r})-\asw(w,b,\bs{r})+\sum_{t=0}^T w_i \E_{\pi^*(w,b,\bs{r})}[r_i(s_t,a_t)]\right)
\end{align*}
where $\bs{r}=(r_i,r_{-i})$.
The true expected utility at the time they report, as a function of their true reward function and a possible misreport $\Tilde{r}_i$, is
\begin{equation*}
    U_i(r_i, \hat{r}_i) = \mathbb{E}_{\pi^*(w,b,\Tilde{r}_i,r_{-i})}\left[\sum_{t=0}^T r_i(s_t, a_t) - \frac{1}{w_i}\left(\asw^{(-i)}(w, b, \bs{r}) - \left(\sum_{t=0}^T\left(\sum_{j\neq i}  w_j r_j(s_t,a_t)\right)+ b(s_t,a_t)\right)\right)\right]
\end{equation*}
where the expectation is over the randomness in the MDP (note that it does \textit{not} need to be over the randomness in opponent types -- IC should hold in dominant strategies considering the opponents).

In choosing a misreport $\tilde{r}_i$, player $i$ thus faces a maximization problem:
\begin{align*}
    \arg\max_{\tilde{r}_i} U_i(r_i, \tilde{r}_i) &= \arg\max_{\tilde{r}_i}  \mathbb{E}_{\pi^*(w,b,\Tilde{r}_i,r_{-i})}\Biggl[\sum_{t=0}^T r_i(s_t, a_t) - \frac{1}{w_i}\Biggl(\asw^{(-i)}(w, b, \bs{r}) \\ &\quad\qquad\qquad - \Biggl(\sum_{t=0}^T\Biggl(\sum_{j\neq i}  w_j r_j(s_t,a_t)\Biggr)+ b(s_t,a_t)\Biggr)\Biggr)\Biggr]\\
    &= \arg\max_{\tilde{r}_i} \mathbb{E}_{\pi^*(w,b,\Tilde{r}_i,r_{-i})}\left[\sum_{t=0}^T r_i(s_t, a_t) +  \frac{1}{w_i}\left(\sum_{t=0}^T \sum_{j \neq i} w_j r_j(s_t, a_t) + b(s_t, a_t)\right)\right] - c \\
        &= \arg\max_{\tilde{r}_i} \mathbb{E}_{\pi^*(w,b,\Tilde{r}_i,r_{-i})}\left[\sum_{t=0}^T w_i r_i(s_t, a_t) +  \sum_{t=0}^T \sum_{j \neq i} w_j r_j(s_t, a_t) + b(s_t, a_t)\right]\\
\end{align*}
where $c$ is a constant.
This is exactly the objective that the mechanism attempts to maximize if the bidder reports truthfully, so the best choice they can make is to do so.
\end{proof}

\subsection{IR and IC guarantees}

In mechanism design, there is often a distinction between:
\begin{itemize}
    \item \textit{Ex post} -- properties that hold after all types have been reported and decisions have been made
    \item \textit{Ex interim} -- properties that hold after a bidder has observed their own type, but before seeing other types.
    \item \textit{Ex ante} -- properties that hold before types have been observed.
\end{itemize}

As far as the prior distribution of agent types is concerned, all our guarantees are \textit{ex post} -- i.e. we ensure dominant-strategy incentive compatibility and ex-post IR.

In our problem, we have an additional source of randomness, the inherent randomness in the MDP itself.
We thus refer to ``in expectation'' to refer to properties that hold when averaging over the randomness in the MDP, but ex post in the types.

Given the above choices of terminology, and given a correct policy and value estimate, our chosen mechanism and payment rules can guarantee in-expectation incentive compatibility and individual rationality.

\section{Proof of  Theorem \ref{thm:ldiff}}
\label{app:proofdiff}
First, let us note that by Rademacher's theorem differentiability almost everywhere follows from local Lipschitz continuity of a function. Moreover, the expected value and derivative are linear operators. Together this implies we only need to show that the expected values of $\asw$ and $g_k$ are locally Lipschitz continuous.

In this proof, we identify boosts and agents' rewards with vectors in $\mathbb{R}^{|S| \times |A|}$ and weights with vectors in $\mathbb{R}^n_{+}$. Moreover, we will assume that $f$ is a continuous density function, according to which $\bs{r}$ is distributed and that it is sufficiently well-behaved, such that $\mathbb{E}_{\bs{r}}[\norm{\bs{r}}_\infty]$ exists.

\subsection{Affine social welfare is Locally Lipschitz continuous}
Here we only show continuity in $w$, since the proof for $b$ works analogous. Fix $\bs{r}\in \R^{|S|\times|A|\times n}$ and let $w,\hat{w} \in \R_{+}^n$, s.t. $\exists ! j: w_j \neq \hat{w}_j$. We need to show that $\exists K_{\bs{r}}$, s.t.

\begin{dmath}
\begin{split}  
| \left(\sum_{i=1}^n\sum_{t=0}^T w_i \E_{s_t,a_t \sim \pi}[r_i(s_t,a_t)]\right) + \sum_{t=0}^T \E_{s_t,a_t \sim \pi}[b(s_t,a_t)]
    -  \left(\sum_{i=1,i\neq j}^n\sum_{t=0}^T w_i \E_{s_t,a_t \sim \hat{\pi}}[r_i(s_t,a_t)]\right) \\ - \sum_{t=0}^T \hat{w}_j  \E_{s_t,a_t \sim \hat{\pi}}[r_j(s_t,a_t)] - \sum_{t=0}^T \E_{s_t,a_t \sim \hat{\pi}}[b(s_t,a_t)] | \leq K_{\bs{r}} |w_j-\hat{w}_j|
\end{split}
\end{dmath}
where $\pi=\pi^*(w,b,\bs{r})$ and $\hat{\pi}=\pi^*(\hat{w},b,\bs{r})$.
To perform our proof we need the following inequalities:

By optimality of $\pi$
\begin{equation}
\label{eq:piopt}
    \begin{split}
        \sum_{i=1,i\neq j}^n(\sum_{t=0}^T w_i \E_{s_t,a_t \sim {\pi}}[r_i(s_t,a_t)]) + \sum_{t=0}^T {w}_j  \E_{s_t,a_t \sim {\pi}}[r_j(s_t,a_t)] + \sum_{t=0}^T \E_{s_t,a_t \sim {\pi}}[b(s_t,a_t)]
    \\ \geq  \sum_{i=1,i\neq j}^n(\sum_{t=0}^T w_i \E_{s_t,a_t \sim \hat{\pi}}[r_i(s_t,a_t)]) + \sum_{t=0}^T {w}_j  \E_{s_t,a_t \sim \hat{\pi}}[r_j(s_t,a_t)] + \sum_{t=0}^T \E_{s_t,a_t \sim \hat{\pi}}[b(s_t,a_t)]
    \end{split}
\end{equation}

By optimality of $\hat{\pi}$
\begin{equation}
\label{eq:pihatopt}
    \begin{split}
        \sum_{i=1,i\neq j}^n(\sum_{t=0}^T w_i \E_{s_t,a_t \sim {\pi}}[r_i(s_t,a_t)]) + \sum_{t=0}^T \hat{w}_j  \E_{s_t,a_t \sim {\pi}}[r_j(s_t,a_t)] + \sum_{t=0}^T \E_{s_t,a_t \sim {\pi}}[b(s_t,a_t)]
    \\ \leq  \sum_{i=1,i\neq j}^n(\sum_{t=0}^T w_i \E_{s_t,a_t \sim \hat{\pi}}[r_i(s_t,a_t)]) + \sum_{t=0}^T \hat{w}_j  \E_{s_t,a_t \sim \hat{\pi}}[r_j(s_t,a_t)] + \sum_{t=0}^T \E_{s_t,a_t \sim \hat{\pi}}[b(s_t,a_t)]
    \end{split}
\end{equation}

By assumption on $w,\hat{w_j}$

\begin{equation}
\label{eq:wminushatw}
    \begin{split}
       \sum_{i=1,i\neq j}^n(\sum_{t=0}^T w_i \E_{s_t,a_t \sim {\pi}}[r_i(s_t,a_t)]) + \sum_{t=0}^T {w}_j  \E_{s_t,a_t \sim {\pi}}[r_j(s_t,a_t)] + \sum_{t=0}^T \E_{s_t,a_t \sim {\pi}}[b(s_t,a_t)]
    \\ - \left(  \sum_{i=1,i\neq j}^n(\sum_{t=0}^T w_i \E_{s_t,a_t \sim {\pi}}[r_i(s_t,a_t)]) + \sum_{t=0}^T \hat{w}_j  \E_{s_t,a_t \sim {\pi}}[r_j(s_t,a_t)] + \sum_{t=0}^T \E_{s_t,a_t \sim {\pi}}[b(s_t,a_t)] \right) \\
    \leq |w_j-\hat{w}_j|\sum_{t=0}^T \E_{s_t,a_t \sim {\pi}}[r_j(s_t,a_t)] 
    \end{split}
\end{equation}

\begin{equation}
\label{eq:wminushatwpihat}
    \begin{split}
       \sum_{i=1,i\neq j}^n(\sum_{t=0}^T w_i \E_{s_t,a_t \sim {\hat{\pi}}}[r_i(s_t,a_t)]) + \sum_{t=0}^T \hat{w}_j  \E_{s_t,a_t \sim {\hat{\pi}}}[r_j(s_t,a_t)] + \sum_{t=0}^T \E_{s_t,a_t \sim {\hat{\pi}}}[b(s_t,a_t)] 
    \\ - \left(  \sum_{i=1,i\neq j}^n(\sum_{t=0}^T w_i \E_{s_t,a_t \sim {\hat{\pi}}}[r_i(s_t,a_t)]) + \sum_{t=0}^T {w}_j  \E_{s_t,a_t \sim {\hat{\pi}}}[r_j(s_t,a_t)] + \sum_{t=0}^T \E_{s_t,a_t \sim {\hat{\pi}}}[b(s_t,a_t)]   \right) \\
    \leq |w_j-\hat{w}_j|\sum_{t=0}^T \E_{s_t,a_t \sim {\hat{\pi}}}[r_j(s_t,a_t)] 
    \end{split}
\end{equation}

Now we have all the ingredients to show local Lipschitz continuity wrt to $w_i$. Let $K_{\bs{r}}= \norm{\bs{r}}_{\infty}T\geq\max(\sum_{t=0}^T \E_{s_t,a_t \sim {\hat{\pi}}}[r_j(s_t,a_t)],\sum_{t=0}^T \E_{s_t,a_t \sim {\pi}}[r_j(s_t,a_t)] )$. Then we have
\begin{align*}
     & \sum_{i=1,i\neq j}^n(\sum_{t=0}^T w_i \E_{s_t,a_t \sim {{\pi}}}[r_i(s_t,a_t)]) + \sum_{t=0}^T {w}_j  \E_{s_t,a_t \sim {{\pi}}}[r_j(s_t,a_t)] + \sum_{t=0}^T \E_{s_t,a_t \sim {{\pi}}}[b(s_t,a_t)]+ K_{\bs{r}} |w_j-\hat{w}_j| \\
     & \underset{Eq.\eqref{eq:piopt}}{\geq} \sum_{i=1,i\neq j}^n(\sum_{t=0}^T w_i \E_{s_t,a_t \sim {\hat{\pi}}}[r_i(s_t,a_t)]) + \sum_{t=0}^T {w}_j  \E_{s_t,a_t \sim {\hat{\pi}}}[r_j(s_t,a_t)] + \sum_{t=0}^T \E_{s_t,a_t \sim {\hat{\pi}}}[b(s_t,a_t)]+ K_{\bs{r}} |w_j-\hat{w}_j| \\
     & \underset{Eq.\eqref{eq:wminushatwpihat}}{\geq}  \sum_{i=1,i\neq j}^n(\sum_{t=0}^T w_i \E_{s_t,a_t \sim {\hat{\pi}}}[r_i(s_t,a_t)]) + \sum_{t=0}^T \hat{w}_j  \E_{s_t,a_t \sim {\hat{\pi}}}[r_j(s_t,a_t)] + \sum_{t=0}^T \E_{s_t,a_t \sim {\hat{\pi}}}[b(s_t,a_t)]\\
     & \underset{Eq.\eqref{eq:pihatopt}}{\geq}  \sum_{i=1,i\neq j}^n(\sum_{t=0}^T w_i \E_{s_t,a_t \sim {{\pi}}}[r_i(s_t,a_t)]) + \sum_{t=0}^T \hat{w}_j  \E_{s_t,a_t \sim {{\pi}}}[r_j(s_t,a_t)] + \sum_{t=0}^T \E_{s_t,a_t \sim {{\pi}}}[b(s_t,a_t)] \\
     & \underset{Eq.\eqref{eq:wminushatw}}{\geq} \sum_{i=1,i\neq j}^n(\sum_{t=0}^T w_i \E_{s_t,a_t \sim {{\pi}}}[r_i(s_t,a_t)]) + \sum_{t=0}^T {w}_j  \E_{s_t,a_t \sim {{\pi}}}[r_j(s_t,a_t)] + \sum_{t=0}^T \E_{s_t,a_t \sim {{\pi}}}[b(s_t,a_t)]- K_{\bs{r}} |w_j-\hat{w}_j| \\
\end{align*}
To conclude this proof, we note $$\E_{\bs{r}}[\asw(w,b,\bs{r})-\asw(\hat{w},b,\bs{r})]\leq T\E_{\bs{r}}[\norm{\bs{r}}_\infty]|w_j-\hat{w}_j|
$$
where $\E_{\bs{r}\sim f}[\norm{\bs{r}}_\infty]$ is finite by assumption.

\subsection{$g_k$ is locally Lipschitz continuous}
Unlike for $\asw$, we cannot argue that $g_k$ is Lipschitz continuous. Indeed changing the weights and boosts only slightly can cause a completely different policy to become optimal, leading to a discontinuous jump or drop in $\pi^*(w,b,\bs{r})$ and thus in general also in $g_k$. However, as we will show, when taking the expected value, these discontinuities get smoothed out, guaranteeing local Lipschitz continuity and thereby differentiability almost surely.

We will restrict to proving local Lipschitz continuity of $\E_{\bs{r}}\left[g_k(\pi^*(w,b,\bs{r}),\bs{r})\right]$ with respect to $b$, as the proof with respect to $w$ works similar.

In the proofs below we identify the reported reward functions of all agents with vectors $\bs{r}\in \R^d$, where $d={|S|\times|A|\times n}$. 

Our first observation is that all vectors $\bs{r}$ for which two policies $\pi_1,\pi_2$ give the same affine social welfare lie on a hyperplane. 
Indeed, let $\nu_i$ denote the induced state-action occupancy measure of policy $\pi_i$. Then
$$
\sum_{s,a} \nu_1(s,a) (w_i r_i(s,a) + b(s,a))=\sum_{s,a} \nu_2(s,a) (w_i r_i(s,a) + b(s,a))
$$
is equivalent to 
$$
\sum_{s,a,i}(\nu_1(s,a)-\nu_2(s,a))w_i r_i(s,a) = \sum_{s,a} (\nu_2(s,a)-\nu_1(s,a)) (b(s,a))
$$
Let $\boldsymbol{\nu}_1=(w_1\nu_1(s_1,a_1),w_2\nu_1(s_1,a_1),\dots,w_n \nu_1(s_1,a_1),w_1 \nu_1(s_2,a_1),\dots ,w_n \nu_1(s_{|S|},a_{|A|})$. Then the above is equivalent to the following hyperplane:
$$
\mathcal{H}_{12}(w,b)=\{\boldsymbol{r}:(\boldsymbol{\nu}_{1} - \boldsymbol{\nu}_{2}
)^T \boldsymbol{r} =\sum_{s,a} (\nu_{2}(s,a)-\nu_{1}(s,a)) (b(s,a))
\}$$

Let $b,\Tilde{b}\in \R^{|S|\times|A|}$, s.t. $\exists ! (s',a'): b(s',a') \neq \Tilde{b}(s',a')$. $\Tilde{b}$ gives us another hyperplane of equivalence between $\pi_1,\pi_2$

$$
\mathcal{H}_{12}(w,\Tilde{b})=\{\bs{r}:(\boldsymbol{\nu_1}-\bs{\nu_2})^T \boldsymbol{r} =\sum_{s,a} (\nu_2(s,a)-\nu_1(s,a)) (\tilde{b}(s,a))
\}$$

Note that $\mathcal{H}_{12}(w,\Tilde{b})$ and $\mathcal{H}_{12}(w,{b})$ are parallel with a distance $  \frac{|(\nu_2 (s',a') - \nu_1 (s',a'))(b(s',a')-\Tilde{b}(s',a'))|}{ \norm{ (\bs{\nu_1} - \bs{\nu_2})}}$.

Moreover, for any given parameters $(w,b)$ there are only $|\Pi|^2$ planes of equivalence, where $|\Pi|$ is the number of deterministic policies. 

When we change $b(s',a')$ to $\tilde{b}(s',a')$, then there can be rewards $\bs{r}$, for which $\pi^*(w,\tilde{b},\bs{r}) \neq \pi^*(w,b,\bs{r})$. For these rewards it follows that in general  $g_k(\pi^*(w,\Tilde{b},\bs{r}),\bs{r})$ is different from $g_k(\pi^*(w,b,\bs{r}),\bs{r})$. Using the hyperplanes defined above, we know the set of all $\bs{r}$, where $g_k(\pi^*(w,\Tilde{b},\bs{r}),\bs{r}) \neq g_k(\pi^*(w,b,\bs{r}),\bs{r})$ is contained in the set

\begin{equation*}
    U(w,b,\Tilde{b})= \bigcup\limits_{\pi_j,\pi_i \in \Pi}U(w,b,\Tilde{b})_{ij}
\end{equation*}
where 
\begin{align*}
U(w,b,\Tilde{b})_{ij}=\{\boldsymbol{r}: &|(\bs{\nu_i}-\bs{\nu_j})^T \boldsymbol{r} - \sum_{s,a} (\nu_j(s,a)-\nu_i(s,a)) (b(s,a))| \leq \\&|(\nu_j(s',a')-\nu_i(s',a')) (\tilde{b}(s',a')-b(s',a'))|\}
\end{align*}
are the polytopes induced by the hyperplanes $\mathcal{H}_{ij}(w,\tilde{b})$ and $\mathcal{H}_{ij}(w,{b})$.

With this in mind let us make a first naive analysis of the difference of the change in expectation of $g_k$, when changing $b$.

\begin{align}
    &|\E_{\boldsymbol{r}}[g_k(\pi^*(w,\Tilde{b},\bs{r}),\bs{r})] - \E_{\boldsymbol{r}}[g_k(\pi^*(w,{b},\bs{r}),\bs{r})]|\\
    &\leq E_{\boldsymbol{r}}[|g_k(\pi^*(w,\Tilde{b},\bs{r}),\bs{r}) -g_k(\pi^*(w,{b},\bs{r}),\bs{r})| ] \\
    &\leq  C \int_{U(w,b,\Tilde{b})} \norm{\bs{r}}_{\infty}f(\boldsymbol{r}) d\boldsymbol{r}\\
    &\leq C \sum_{\pi_i,\pi_j \in |\Pi|}   \int_{U(w,b,\Tilde{b})_{ij}} \norm{\bs{r}}_{\infty}f(\boldsymbol{r}) d\boldsymbol{r}\\
\end{align}
for some constant $C$ (since $g_k = \mathcal{O}(\norm{\bs{r}}_{\infty})$). 

We want to show that the above can be bounded by $L |b(s',a')-\tilde{b}(s',a')|$ for some $L$. For this we need to get a better understanding of $\int_{U(w,b,\Tilde{b})_{ij}} \norm{\bs{r}}_{\infty}f(\boldsymbol{r}) d\boldsymbol{r}$. 

For the sake of simplicity, we assume now the probability density $f$ has compact support on $R^d$,  i.e. there exists a $K$ such that $\forall \bs{r}: \norm{\bs{r}}_{2}\geq K \implies f(r)=0$.\footnote{This assumption is not necessary. As long as $f$ decays sufficiently quickly for large $\bs{r}$, the proof still goes through with some minor adjustments. However, we make the assumption here for streamlining our exposition and highlighting the parts of our proof, which are non-standard.} Since the hyperplanes  $\mathcal{H}_{ij}(w,\Tilde{b})$ and $\mathcal{H}_{ij}(w,b)$ have dimension $d-1$ and are parallel with distance $  \frac{|(\nu_2(s',a')-\nu_1(s',a'))(b(s',a')-\Tilde{b}(s',a'))|}{ \norm{ (\bs{\nu_1}-\bs{\nu_2})}}$, we can bound the integral by multiplying an upper bound of the volume of $U(w,b,\Tilde{b})_{ij}$ with the maximum possible value of $\norm{\bs{r}}_{\infty}f(\boldsymbol{r})$ .

\begin{align*}
    \int_{U(w,b,\Tilde{b})_{ij}} \norm{\bs{r}}_{\infty}f(\boldsymbol{r}) d\boldsymbol{r} &\leq (2K)^{d-1} \frac{|(\nu_2(s',a')-\nu_1(s',a'))(b(s',a')-\Tilde{b}(s',a'))|}{ \norm{ (\bs{\nu_1}-\bs{\nu_2})}} K \max_{\bs{r}}f(\bs{r})\\
    =&L|(b(s',a')-\Tilde{b}(s',a'))|
\end{align*}
for $L=(2K)^{d} \frac{|(\nu_2(s',a')-\nu_1(s',a'))|}{ \norm{ (\boldsymbol{\nu}_1-\boldsymbol{\nu}_2)}} \max_{\bs{r}}f(\bs{r})$, which proves Lipschitz continuity and thereby differentiability almost surely.

\section{Gradients of affine social welfare}
\label{app:grad_asw}

As outlined in Section \ref{sec:lpformulation}, if we can take the derivative inside the expected value, the implicit function theorem yields:
\[
\nabla_{w, b} \mathbb{E}_{\boldsymbol{r}}\left[\mathcal{L}\left(\pi^*, w, b\right)\right]=\mathbb{E}_{\boldsymbol{r}}\left[\nabla_2 \mathcal{L}\left(\pi^*, w, b\right)+\nabla_{w, b} \pi^*(w, b, \boldsymbol{r}) \nabla_1 \mathcal{L}\left(\pi^*, w, b\right)\right]
\]
Here we show that we can analytically compute the partial derivatives with respect to $w,b$ (keeping a policy $\pi$ fixed) of $\asw$, $\sw$ (and thus also for revenue), as well as for makespan, which is important given that we can use these to accelerate the gradient computation for our second approach where we take the gradient through the regularised LP.
We note that for $ \asw(w,b,r)$ the gradients can be computed in a straight-forward manner.\footnote{The analysis holds equivalently for $\asw_{-i}(w,b,r)$} 
For this, rewrite $\asw$ using the state-action occupancy measure $\nu_{\pi}$. \footnote{In general finite horizon MDPs this would be defined as $\nu_{\pi}(s,a)=\sum_{t=1}^T \mathbb{P}_{\pi}(s_t=s,a_t=a)$. In our experiments we always assume the states contain the current timestep such that this simplifies to $\nu_{\pi}(s,a)=\mathbb{P}_{\pi}(s_t=s,a_t=a)$. Note that the same analysis equivalently holds for $\asw^{(-i)}$}

$$
\asw(\pi,w,b,\bs{r})= \sum_{s,a} \nu_\pi(s,a) (\sum_{i=1}^n w_i r_i(s,a) + b(s,a))
$$
In this form taking the partial derivative is straightforward. We get the following
\begin{equation*}
\nabla_{w,b} \asw(w,b,r) = \nabla_{w,b}\sum_{s,a} \nu_{\pi}(s,a) (\sum_{i=1}^n w_i r_i(s,a) + b(s,a))
\end{equation*}
\begin{equation*}
	\nabla_{w_i} \asw(w,b,r) = \sum_{s,a} \nu_{\pi}(s,a)r_i(s,a)=\E_{\pi}[\sum_{t=0}^T  r_i(s,a)] 
 \end{equation*}
 \begin{equation*}
	\nabla_{b(s,a)} \asw(w,b,r) = \nu_{\pi}(s,a)
\end{equation*}
Boosting a state increases $\asw$ in relation to how often the state is visited under the optimal policy. This is easy to compute. In fact, $\nu_{\pi}$ is the solution to the linear programming formulation of the MDP. Similarly increasing the weight of an agent changes $\asw$ in proportion to the expected sum of rewards the agent gets. 

For $\sw$ and makespan the analysis is even simpler, since the partial derivative with respect to $w,b$, keeping a policy fixed is just 0. As both only depend on these parameters indirectly through the policy.

\section{Proof of Theorem \ref{thm:regconv}}
\label{app:reg_rev}

\begin{proof}
    Fix $w,b,\bs{r}$ and denote by $L_{\alpha}(w,b,\bs{r}) ,L(w,b,\bs{r})$ the optimal regularised and unregularised loss for this specific choice of variables. We first show 
    \[
    \lim_{\alpha\rightarrow 0} L_{\alpha}(w,b,\bs{r}) = L(w,b,\bs{r})
    \]
    
    Let $\nu^*,\nu^{\alpha}$ be the corresponding state-action measures to $\pi^*,\pi^\alpha$---the optimal policies for the unregularised and regularised LP. Using Corollary 9 of \citet{Weed2018Explicit}, for sufficiently small $\alpha$, we get:
    
    \begin{align*}
        & \left |L(w,b,\bs{r})-L_{\alpha}(w,b,\bs{r})\right |\\
        &= \left |\sum_{s,a} (\nu^*(s,a)-\nu^{\alpha}(s,a))\left(\bs{q}\right)\right |\\
        &= \left |\langle\nu^*-\nu^{\alpha}, \bs{q} \rangle \right |\\
        &\leq \norm{\nu^*-\nu^{\alpha}}_{1} \norm{q}_{\infty}\\
        &\leq 2 R_1 \exp\left(-\frac{\Delta(\bs{r})}{\alpha R_1} + \frac{R_1+R_H}{R_1}\right) (C\norm{\bs{r}}_{\infty})\\
    \end{align*}
    where $C$ is some constant, $R_1$ is the $l_1$ radius of all feasible solutions, $R_H$ is the entropic radius, and $\Delta$ the suboptimality gap~\cite{Weed2018Explicit}. This shows pointwise convergence of $L_{\alpha}(w,b,\bs{r})$.
    Note further that for any policy $\pi$ and any $w,b$, $\mathcal{L}_{\pi}(w,b,\bs{r})$ can be bounded by $(C\norm{\bs{r}}_{\infty})$ for some $C$, which by assumption is integrable. It follows by the Domianted Convergence Theorem that 
    \[\forall w,b:\quad \lim_{\alpha\rightarrow 0}\E\left[\mathcal{L}_{\alpha}(w,b,\bs{r})\right]  = \E\left[\mathcal{L}(w,b,\bs{r})\right]\]
 \end{proof}

\section{Linear Program of regularized MDP}
\label{app:reglp}

Below we give the regularized form of the MDP linear program -- it is now a convex (exponential cone) program.

\begin{equation}
\begin{split}
&	 \max \sum_{s \in S, a \in A} \left( \sum_i w_i r_i(s,a) 
	+ b(s,a) \right) \nu(s, a) + \\
&\quad  \quad \alpha \sum_{s \in S, a \in A} \nu(s,a) \log \nu(s,a) \quad  \text{ s.t.  }\\
	&\sum_{a \in A} \nu(s, a) = \sum_{s',a'} P(s|s', a') \nu(s', a')  + \mu_0(s) \quad \forall s \in S \\
	& \nu(s, a) \geq 0 \quad \forall s \in S, a \in A
 \end{split}
 \label{eq:smoothlp}
\end{equation}

\section{Computational Details and Hyperparameters}
\label{app:hparam}

The grid search experiments and all experiments in the gridworld environment were run concurrently on a server with 256 cores and 250GB of RAM, while restricting the number of threads in MOSEK to 4.
Other experiments were run on cluster nodes with 4 cores and 1GB or 2GB of RAM per core, except that the task scheduling regularized LP jobs with 3 agents were run with 16 cores and 64GB memory. During development we experimented with up to 1000 sampled valuation profiles, up to 2000 perturbations, learning rates ranging from 0.001 to 0.1, and regularization strengths up to 0.1; we quickly settled on the chosen hyperparameters and did not do a more exhaustive search due to computational constraints. For distributions where the bidder valuations are symmetric, we optimize only boosts, fixing the weights to 1.

\end{document}